\documentclass[11pt,a4paper]{article}

\usepackage[margin=1in]{geometry}
\usepackage[T1]{fontenc}
\usepackage[utf8]{inputenc}
\usepackage{lmodern}
\usepackage{microtype}
\usepackage{graphicx}
\usepackage{subcaption}
\usepackage{booktabs}
\usepackage{multirow}
\usepackage{amsmath,amssymb,amsthm}
\usepackage{array}
\usepackage{enumitem}
\usepackage{hyperref}
\usepackage{float}
\usepackage{authblk}
\usepackage{siunitx}
\usepackage[backend=biber,style=numeric,sorting=none,maxbibnames=99]{biblatex}
\hypersetup{colorlinks=true,linkcolor=blue,citecolor=blue,urlcolor=blue}
\newtheorem{proposition}{Proposition}
\title{Local SVD-Entropy Maps as a Complementary Structural Representation for Full-Reference and No-Reference Image Quality Assessment}

\author[1]{Andrei Velichko\thanks{Corresponding author: \href{mailto:velichkogf@gmail.com}{\texttt{velichkogf@gmail.com}}; ORCID: \href{https://orcid.org/0000-0002-9341-1831}{0000-0002-9341-1831}}}
\author[1]{Petr Boriskov\thanks{E-mail: \href{mailto:boriskov@petrsu.ru}{\texttt{boriskov@petrsu.ru}}; ORCID: \href{https://orcid.org/0000-0002-2904-9612}{0000-0002-2904-9612}}}
\affil[1]{Institute of Physics and Technology, Petrozavodsk State University, 185910 Petrozavodsk, Russia}
\date{}

\begin{document}
\maketitle

\begin{abstract}
We investigate a local spectral-complexity representation for perceptual image quality assessment (IQA) based on Shannon entropy of singular values computed directly from two-dimensional image patches. For each $3\times3$-pixel grayscale patch, SVD is applied directly and the normalized singular-value entropy defines one HSVD-map value. The construction requires neither flattening nor delay embedding, uses no boundary padding, and is invariant to $90^{\circ}$ rotations and mirror reflections at the local-descriptor level.

A nested salt-and-pepper experiment on Lena separates absolute similarity to a clean reference from sensitivity to an additional degradation step. HSVD-SSIM responds more strongly to local corruption and retains a larger neighboring-state response at severe noise levels. Validation on all 10,125 distorted KADID-10k images shows that HSVD-SSIM is weaker than conventional SSIM as a standalone full-reference metric (SRCC $0.450$ vs. $0.619$), but complementary when combined with it: grouped cross-validation increases SRCC from $0.618$ to $0.659$, with a bootstrap 95\% confidence interval of $[0.036,0.046]$ for the gain. In a no-reference experiment, adding HSVD-derived single-image descriptors improves the best nonlinear model from SRCC $0.528$ to $0.575$ (95\% CI $[0.033,0.061]$) and also improves prediction of quality changes between neighboring distortion states.

These results support direct local SVD entropy as an interpretable structural channel that complements conventional image-domain similarity and remains informative without a pristine reference.
\end{abstract}


\section{Introduction}
Perceptual image quality assessment (IQA) seeks objective measures that reproduce, as closely as possible, human judgments of visual quality. Depending on the availability of reference information, IQA methods are commonly divided into full-reference (FR), reduced-reference (RR), and no-reference (NR) approaches \cite{Hosu2020KonIQ,Dost2022RRReview,Zhang2019BlindIQA,Liu2023DeepBlind}. FR methods are advantageous when a pristine image is available, whereas NR methods address the more difficult and practically important situation in which quality must be inferred from the observed image alone. Despite the rapid development of deep IQA models, subjective opinion studies remain the perceptual ground truth used to train and benchmark objective measures, and large-scale databases together with content-separated evaluation protocols remain essential for assessing generalization \cite{Hosu2020KonIQ,Min2024Subjective,Frackiewicz2024Combined}. Objective IQA therefore remains relevant for image acquisition, compression, transmission, restoration, enhancement, and automated image-processing systems in which repeated human evaluation is impractical.

The Structural Similarity Index (SSIM) became one of the most influential FR-IQA approaches by shifting the emphasis from pixelwise error visibility toward preservation of local luminance, contrast, and structure \cite{Wang2004SSIM}. Its multiscale extension (MS-SSIM) incorporated structural comparison across resolutions \cite{Wang2003MSSSIM}, while information-theoretic and feature-based alternatives such as VIF, information-weighted SSIM, and FSIM demonstrated that perceptual fidelity can benefit from cues beyond a single local structural statistic \cite{Sheikh2006VIF,Wang2011IWSSIM,Zhang2011FSIM}. SSIM nevertheless remains an important and interpretable baseline because of its simplicity, locality, and direct relation to spatial structure.

SSIM is not a complete model of visual perception. Its behavior depends on implementation choices, local windowing, pooling, image content, and distortion regime \cite{Venkataramanan2021SSIMGuide,Bakurov2021SSIM}. Task-specific variants have therefore been proposed for intensity-sensitive or specialized imaging conditions \cite{Armour2022IntensitySSIM,Milovic2024XSIM}, while spatial misalignment has motivated comparisons in alternative feature spaces \cite{Liao2024OrderStats}. Medical-imaging studies likewise show that the most informative quality measure can depend on the application \cite{Mason2020MRI}. These results motivate complementary representations when a single image-domain score does not capture all perceptually relevant manifestations of degradation.

This premise has led to a well-established line of \emph{complementary} and multi-domain IQA models. FSIM combines phase congruency and gradient magnitude because the two features describe different perceptual properties \cite{Zhang2011FSIM}; MDSI combines gradient and chromaticity information \cite{Nafchi2016MDSI}; spatial-frequency and spatial-transformed-domain methods fuse features extracted from distinct representations \cite{Tang2019SpatialFrequency,Yu2020MultiFeature}; and optimization-based approaches explicitly combine multiple quality measures to increase agreement with subjective scores \cite{Varga2022MetricFusion,Varga2023FusionModels}. Genetic-programming and dual-space methods provide further evidence that the quality signal available from different metrics or representation spaces need not be fully redundant \cite{Bakurov2023GP,Wu2024DualSpace}. The same principle appears in learned representations: DISTS unifies structure and texture similarity in deep feature space \cite{Ding2020DISTS}, while recent blind-IQA work has shown that shape- and texture-biased deep statistics contribute differently to quality prediction \cite{Li2024ShapeTexture}. Thus, complementarity rather than standalone superiority is a well-established strategy in modern IQA.

The problem of complementary IQA is also related to the distinction between \emph{absolute quality} and \emph{sensitivity to a change between two image states}. Pairwise and triplet studies can resolve small quality differences that are difficult to express on coarse absolute scales \cite{Men2021Pairwise}, while restoration-oriented benchmarks such as PIPAL show that denoising, super-resolution, and generative restoration can create perceptually important artifacts not summarized by traditional distortion measures alone \cite{Gu2020PIPAL}. A representation may therefore be highly responsive to an additional local degradation step without being the best standalone predictor of absolute perceptual quality. We evaluate these properties separately.

Entropy-based descriptors are useful in this complementary viewpoint because they characterize irregularity, complexity, or information distribution rather than directly comparing pixel intensities. Spatial and spectral entropy features have already been used for NR-IQA \cite{Liu2014EntropyIQA}, and local entropy representations have been related to perceptual quality and spatial image structure \cite{Chen2019EntropyIQA}. More recent work continues to combine local information entropy with learned IQA models \cite{Zhang2025EntropyViT}. At a broader image-analysis level, explicitly two-dimensional entropy formulations have been developed to retain spatial relationships that can be weakened when an image neighborhood is reduced to a one-dimensional sequence \cite{Abid2025IncrEn2D}.

Our own earlier work followed this latter route. Velichko et al. constructed two-dimensional entropy distributions by scanning local circular image kernels, converting each two-dimensional neighborhood into a specially ordered one-dimensional sequence, and then evaluating SVD, sample, permutation, or neural-network entropy on the resulting series \cite{Velichko2022EntropyApprox}. The circular traversal was introduced specifically to improve stability of the entropy result after image rotation. For SvdEn2D, the resulting one-dimensional series was subsequently embedded before singular-value entropy was evaluated. This approach demonstrated that spatial entropy maps can provide useful irregularity representations, but it still required a designed traversal rule and an intermediate sequence/embedding construction.

The formulation studied here is a methodological simplification of that earlier idea. Instead of attempting to design a scan order that is approximately insensitive to orientation, we remove the scan order altogether and compute SVD directly from the raw two-dimensional local patch. The normalized Shannon entropy of the patch singular values then defines one scalar local-complexity value. Because the singular-value spectrum is unchanged by transposition and left/right multiplication by orthogonal permutation matrices, invariance of the local descriptor to $90^{\circ}$ rotations and mirror reflections follows directly from matrix algebra rather than from the choice of a traversal heuristic. The direct formulation is also algorithmically simpler because sequence construction and delay embedding are no longer required. We do not, however, claim a measured speed advantage over the earlier implementation in this paper because the two algorithms have not been benchmarked under identical computational conditions.

SVD itself is well established in image processing. Low-rank and patch-based SVD methods exploit the concentration of structural energy in dominant singular directions for denoising and restoration \cite{Guo2016SVDDenoising}, while SVD entropy has been used to monitor residual rendering noise and image complexity \cite{Buisine2021SVDEntropy}. SVD-derived IQA models include structural-SVD descriptors \cite{Mansouri2019SSVD}, quaternion-SVD formulations for color IQA \cite{Sang2020QuaternionSVD}, and local patch-structure representations \cite{Wang2015PatchStructure}; singular-value-based patch similarity has also been used in denoising \cite{Wang2021SVDPatchSimilarity}. To our knowledge, however, direct local 2-D SVD-entropy maps subsequently compared by SSIM have not been established as an FR-IQA representation. The proposed HSVD map therefore lies between local SVD descriptors, spatial entropy maps, and structural-similarity-based IQA.

The same question is relevant when no pristine reference is available. Classical NR-IQA uses natural-scene statistics, texture, gradients, transform coefficients, entropy, and other interpretable low-dimensional descriptors, and recent feature-fusion methods continue to combine heterogeneous families of such cues \cite{Varga2020NRFusion,Cui2020DualDomainNR,Varga2022LocalDescriptors,Varga2023GlobalLocal}. Singular-value-derived statistics have also been combined with wavelet, DCT, moment, and other global descriptors \cite{Varga2021GlobalStats}. Deep NR-IQA generally achieves stronger headline performance, but hybrid statistical/deep formulations reinforce the broader idea that low- and high-level cues can be complementary \cite{Zhang2019BlindIQA,Ni2024MDFS}. Here the NR experiment is used specifically to test whether HSVD statistics retain useful quality information when the reference image is absent.

Table~\ref{tab:related_intro} summarizes the main literature lines that motivate the present study and clarifies how they differ from the proposed representation.

\begin{table}[H]
\centering
\caption{Representative literature lines motivating the present work and their relation to the proposed HSVD representation.}
\label{tab:related_intro}
\footnotesize
\begin{tabular}{p{2.35cm}p{3.35cm}p{4.0cm}p{4.25cm}}
\toprule
Approach & Representation or cues & Main idea & Relation to the present work \\
\midrule
SSIM / MS-SSIM \cite{Wang2004SSIM,Wang2003MSSSIM} & Local luminance, contrast, structure; multiple scales & Direct structural comparison in the image domain & Baseline structural similarity used throughout the FR experiments \\
Complementary FR-IQA \cite{Zhang2011FSIM,Nafchi2016MDSI,Ding2020DISTS} & Gradient, phase, color, structure, texture, learned features & Combine complementary perceptual cues & Establishes complementarity as a valid alternative to ``replacement'' claims \\
Pairwise / restoration quality \cite{Men2021Pairwise,Gu2020PIPAL} & Relative human judgments; restoration-specific artifacts & Resolve small quality differences and complex restoration degradations & Motivates separate analysis of absolute quality and state-to-state sensitivity \\
Entropy-based IQA \cite{Liu2014EntropyIQA,Chen2019EntropyIQA} & Spatial, spectral, and local entropy & Describe irregularity and local complexity rather than only pixel fidelity & Motivates entropy as an IQA-relevant representation \\
Earlier 2-D entropy maps \cite{Velichko2022EntropyApprox} & Circular 2-D kernel $\rightarrow$ ordered 1-D sequence $\rightarrow$ entropy & Build spatial irregularity maps with rotation-stable traversal & Direct methodological precursor; the present method removes traversal and embedding \\
SVD-based IQA / patch analysis \cite{Mansouri2019SSVD,Sang2020QuaternionSVD,Wang2015PatchStructure} & Singular-value and patch-structure descriptors & Use local or global matrix structure for distortion analysis & Supports singular values as structural features but not the same local HSVD-map construction \\
Interpretable NR-IQA fusion \cite{Varga2020NRFusion,Cui2020DualDomainNR,Varga2023GlobalLocal} & Statistical, perceptual, texture, entropy, transform features & Fuse heterogeneous single-image descriptors & Motivates testing whether HSVD adds useful no-reference information \\
\textbf{Present work} & \textbf{Direct $3\times3$ 2-D SVD-entropy map + SSIM / HSVD statistics} & \textbf{Rotation/reflection-invariant local complexity as an additional structural channel} & \textbf{Absolute FR, incremental, complementary FR, and NR validation on KADID-10k} \\
\bottomrule
\end{tabular}
\end{table}

In the present study we ask whether direct local SVD entropy captures a component of structural degradation that is not fully represented by conventional image-domain similarity. Four linked tests are considered: (1) standalone FR quality prediction; (2) sensitivity to incremental degradation between neighboring image states; (3) complementary information when HSVD-SSIM is combined with conventional SSIM; and (4) no-reference quality prediction from HSVD-derived single-image statistics. KADID-10k is suitable for these tests because it provides 81 source images, 25 distortion types, five severity levels, and subjective quality annotations for 10,125 distorted images \cite{Lin2019KADID}.

The study therefore combines an interpretable proof-of-behavior experiment with large-scale subjective-data validation. Nested salt-and-pepper corruption of Lena is used only to expose the physical response of the two representations: conventional SSIM provides a smoother clean-reference response, whereas HSVD similarity retains stronger sensitivity to successive local changes once both states are already severely degraded. KADID-10k is then used to determine whether this behavior carries perceptual information across many contents and distortions. The analysis includes standalone FR correlations, neighboring-level change correlations, content-grouped cross-validation, reference-level bootstrap resampling, partial/incremental information tests through combined prediction, and an NR experiment in which the pristine image is never used as an input feature.

The main contributions of this work can be summarized as follows. First, we introduce a direct local two-dimensional SVD-entropy map computed from valid $3\times3$-pixel patches without flattening, delay embedding, or artificial boundary padding, and we explicitly establish right-angle rotation and reflection invariance of the local descriptor. Second, we distinguish absolute quality prediction from incremental state-to-state sensitivity and show why these properties should not be conflated. Third, we perform a large-scale KADID-10k validation showing that HSVD-SSIM is weaker than conventional SSIM as a standalone FR metric but provides statistically stable complementary information when the two are combined. Fourth, we demonstrate that HSVD-derived statistics also add useful information in a no-reference setting and improve the recovery of perceptual quality changes from independently predicted image-quality scores. Finally, by comparing the direct formulation with our earlier circular-kernel entropy mapping, we clarify the methodological progression from traversal-designed rotational stability to an exact matrix-invariance property obtained without traversal.

The remainder of the paper is organized as follows. Section~2 defines the direct HSVD map, explains valid-window boundary handling, proves its rotation/reflection invariance, and introduces the FR and NR descriptors. Section~3 describes the nested-noise Lena experiment and the four KADID-10k studies. Section~4 reports the experimental results. Section~5 discusses the interpretation of sensitivity, complementarity, and no-reference utility. Section~6 summarizes limitations and directions for further work, and Section~7 concludes the paper.

\section{Method}
\subsection{Direct two-dimensional local SVD entropy}
Let $X\in\mathbb{R}^{M\times N}$ denote a grayscale image normalized to $[0,1]$. At each valid spatial position, a square patch $P_{ij}\in\mathbb{R}^{w\times w}$ is extracted. Unless otherwise stated, $w=3$. The singular value decomposition of the patch is
\begin{equation}
P_{ij}=U_{ij}\Sigma_{ij}V_{ij}^{\top},
\end{equation}
where the diagonal elements of $\Sigma_{ij}$ are the nonnegative singular values $\sigma_1,\ldots,\sigma_w$.

The singular values are converted into a normalized distribution
\begin{equation}
p_k=\frac{\sigma_k}{\sum_{\ell=1}^{w}\sigma_\ell},\qquad k=1,\ldots,w,
\end{equation}
and the normalized SVD entropy is defined as
\begin{equation}
H_{\mathrm{SVD}}(P_{ij})=-\frac{1}{\ln w}\sum_{k=1}^{w}p_k\ln p_k,
\label{eq:hsvd}
\end{equation}
with the usual convention $0\ln0=0$. Thus $H_{\mathrm{SVD}}\in[0,1]$. A locally rank-one patch has entropy close to zero, whereas a patch whose singular values are more evenly distributed has a larger entropy.

Applying Eq.~\eqref{eq:hsvd} to all locations for which the complete $w\times w$ patch is available yields the HSVD map
\begin{equation}
E(X)=\left(H_{\mathrm{SVD}}(P_{ij})\right)_{i,j}.
\label{eq:hsvdmap}
\end{equation}
No boundary padding is used. Consequently,
\begin{equation}
E(X)\in\mathbb{R}^{(M-w+1)\times(N-w+1)}.
\end{equation}
For example, a $512\times512$ image produces a $510\times510$ HSVD map when $w=3$. This valid-window formulation avoids introducing artificial reflected, replicated, or zero-valued pixels at the image boundary.

\subsection{Rotation and reflection invariance of the local descriptor}
The direct two-dimensional formulation has an invariance property that is lost when a patch is first flattened into an arbitrary one-dimensional scan order.

\begin{proposition}
Let $P\in\mathbb{R}^{w\times w}$ and let $\widetilde{P}$ be obtained from $P$ by a $90^{\circ}$ rotation or a mirror reflection. Then
\begin{equation}
H_{\mathrm{SVD}}(\widetilde{P})=H_{\mathrm{SVD}}(P).
\end{equation}
\end{proposition}

\begin{proof}
A rotation or reflection of a square matrix can be represented by a combination of transposition and left/right multiplication by orthogonal permutation matrices, i.e., $\widetilde P=Q_1 P Q_2$ or $\widetilde P=Q_1 P^{\top}Q_2$. Singular values are invariant under transposition and under multiplication by orthogonal matrices. Therefore $P$ and $\widetilde P$ have the same singular values, the same normalized distribution $\{p_k\}$, and hence the same entropy in Eq.~\eqref{eq:hsvd}.
\end{proof}

Accordingly, the local HSVD value is insensitive to the orientation of the patch under right-angle rotations and reflections. At the image level, rotating the entire image rotates the HSVD map correspondingly while preserving its local values. This property should not be confused with registration invariance: if only one of two compared images is rotated relative to the other, their spatial correspondence changes and SSIM is not expected to remain unchanged.

A second useful property follows from the normalization of the singular values. For any nonzero scalar $c$,
\begin{equation}
H_{\mathrm{SVD}}(cP)=H_{\mathrm{SVD}}(P),
\end{equation}
because all singular values are multiplied by the same factor $|c|$ and the normalized probabilities $p_k$ remain unchanged. Thus the local descriptor is insensitive to uniform multiplicative scaling of a nonzero patch before clipping or other nonlinear operations. It is not, however, invariant to additive luminance offsets or arbitrary nonlinear contrast transformations. This separation between local spectral shape and absolute intensity is one reason to expect HSVD to complement, rather than duplicate, conventional SSIM.

\subsection{Image-domain and entropy-domain SSIM}
For two aligned grayscale images $X$ and $Y$, SSIM is evaluated locally from luminance means $\mu_X,\mu_Y$, standard deviations $\sigma_X,\sigma_Y$, and covariance $\sigma_{XY}$ as \cite{Wang2004SSIM}
\begin{equation}
SSIM(X,Y)=\frac{(2\mu_X\mu_Y+C_1)(2\sigma_{XY}+C_2)}{(\mu_X^2+\mu_Y^2+C_1)(\sigma_X^2+\sigma_Y^2+C_2)},
\label{eq:ssim}
\end{equation}
with local values pooled over the image. Conventional image-domain structural similarity is denoted
\begin{equation}
S_{\mathrm{img}}(X,Y)=SSIM(X,Y).
\end{equation}
Entropy-domain structural similarity uses exactly the same SSIM operator but applies it to the HSVD maps defined in Eq.~\eqref{eq:hsvdmap}:
\begin{equation}
S_{\mathrm{HSVD}}(X,Y)=SSIM\bigl(E(X),E(Y)\bigr).
\label{eq:hsvdssim}
\end{equation}
We refer to Eq.~\eqref{eq:hsvdssim} as \emph{HSVD-SSIM}. In the KADID experiments, both ordinary SSIM and HSVD-SSIM were calculated after converting the reference and distorted images to grayscale. The implementation used Gaussian weighting with $\sigma=1.5$, population covariance normalization (\texttt{use\_sample\_covariance=False}), and data range equal to one.

\subsection{No-reference HSVD descriptors}
In the NR experiments no pristine image is used. Instead, each distorted image $X$ is represented by two feature families motivated by interpretable statistical, texture, entropy, and transform-domain NR-IQA approaches \cite{Varga2020NRFusion,Cui2020DualDomainNR,Varga2021GlobalStats,Varga2022LocalDescriptors,Varga2023GlobalLocal}. The image-domain family includes global intensity statistics and quantiles, a 16-bin intensity histogram and its entropy, gradient statistics, Laplacian energy, horizontal and vertical nearest-neighbor correlations, saturation fractions, and an 8-pixel block-boundary discontinuity measure. The HSVD family applies analogous statistics to the map $E(X)$ defined in Eq.~\eqref{eq:hsvdmap}: mean, standard deviation, quantiles, a 16-bin HSVD histogram and entropy, HSVD-gradient statistics, horizontal and vertical map correlations, and block-boundary discontinuity. Crucially, reference-image identity, distortion type, and distortion level are not used as predictive features.

\section{Experimental design}
\subsection{Illustrative nested salt-and-pepper-noise experiment}
The first experiment is intended to explain the response of the representation rather than to serve as a perceptual benchmark. A $512\times512$ grayscale Lena image was corrupted with nested salt-and-pepper noise at densities $5,10,\ldots,100\%$. A single random permutation of all pixel locations was generated. At each higher noise density, a longer prefix of this same permutation was corrupted, so the set of damaged pixels strictly contained the set from the preceding level. Salt and pepper values alternated between 1 and 0. Thus every step from level $k-1$ to level $k$ represents a controlled addition of new impulse corruption without changing which pixels had already been corrupted.

Two complementary comparisons were made. First, every noisy image was compared with the original clean image:
\begin{equation}
S_{\mathrm{img}}(X_0,X_k),\qquad S_{\mathrm{HSVD}}(X_0,X_k).
\end{equation}
Second, consecutive noise states were compared directly:
\begin{equation}
S_{\mathrm{img}}(X_{k-1},X_k),\qquad S_{\mathrm{HSVD}}(X_{k-1},X_k).
\end{equation}
The first comparison describes absolute departure from the reference, whereas the second isolates sensitivity to each additional $5\%$ degradation step.

\subsection{KADID-10k data and protocol}
The main experiments used 10,125 distorted images from KADID-10k: 81 reference images, 25 distortion types, and five nominal levels per distortion type \cite{Lin2019KADID}. KADID subjective scores were collected with a degradation-category-rating protocol in which a pristine image and its distorted version were shown together; the distorted image was rated on a five-point impairment scale from 1 (very annoying) to 5 (imperceptible), with 30 ratings collected for each distorted image \cite{Lin2019KADID}. Differential Mean Opinion Score (DMOS) generally denotes an aggregate subjective score derived from differential judgments relative to a reference. The released KADID file is named \texttt{dmos.csv} and uses the field name \texttt{dmos}; on the official database page its range is documented as $[1,5]$, with larger values corresponding to higher visual quality \cite{KADIDDatabase}. Because this orientation is opposite to the lower-is-better convention used by some other DMOS databases, we retain the KADID field name but explicitly use the native higher-is-better orientation throughout this paper. All images were converted to grayscale/luminance before calculating either conventional SSIM or HSVD features.

The KADID analysis was organized into four complementary studies.

\paragraph{Study I: standalone full-reference quality.}
For every distorted image $D$ with reference $R$, we computed $S_{\mathrm{img}}(R,D)$ and $S_{\mathrm{HSVD}}(R,D)$ and measured their Spearman (SRCC), Pearson (PLCC), and Kendall correlations with the subjective score.

\paragraph{Study II: neighboring-level sensitivity.}
Within every reference-image/distortion-type chain, adjacent nominal levels were paired sequentially. The transitions $1\rightarrow2$, $2\rightarrow3$, $3\rightarrow4$, and $4\rightarrow5$ yielded 8,100 pairs in total. Metric change magnitude was represented by
\begin{equation}
C_{\mathrm{img}}=1-SSIM(D_k,D_{k+1}),\qquad
C_{\mathrm{HSVD}}=1-SSIM(E(D_k),E(D_{k+1})),
\end{equation}
where $|\Delta DMOS|=|q_{k+1}-q_k|$ and $q_k$ denotes the KADID subjective score at level $k$. For the 1,000 cluster-bootstrap replicates, 81 reference-image identifiers were sampled with replacement from the 81 unique references; every selected identifier contributed all of its neighboring-level pairs, so the within-reference dependence structure was preserved before the correlation difference was recomputed.

\paragraph{Study III: full-reference complementarity.}
Absolute DMOS was predicted with three ordinary least-squares linear models using (i) conventional SSIM only, (ii) HSVD-SSIM only, and (iii) both scores. Evaluation used 10 repetitions of five-fold cross-validation grouped by the 81 reference images; all 125 distorted variants of a given reference were therefore always assigned entirely to training or entirely to testing. Bootstrap uncertainty was estimated from 2,000 cluster resamples using the same procedure described above: 81 reference identifiers were sampled with replacement and all distorted variants belonging to each sampled reference were retained together.

\paragraph{Study IV: no-reference prediction.}
Absolute DMOS was predicted from one distorted image at a time using Image-only, HSVD-only, and Combined feature sets. Ridge regression used standardized predictors and $\ell_2$ regularization with $\alpha=10$. The nonlinear model was a HistGradientBoosting regressor with learning rate $0.05$, 220 boosting iterations, at most 15 leaf nodes, a minimum of 30 samples per leaf, and $\ell_2$ regularization $1.0$; remaining settings followed the scikit-learn implementation \cite{Pedregosa2011ScikitLearn}. Both model families were evaluated using five repetitions of five-fold grouped cross-validation by reference image. The final out-of-fold absolute-quality predictions were also used, without additional fitting, to estimate signed and absolute quality changes between adjacent distortion levels. Reference-level bootstrap resampling used 1,000 cluster samples, again resampling reference identifiers with replacement while retaining all distorted images from a sampled reference together.

\section{Results}
\subsection{Nested noise: absolute similarity and incremental sensitivity}
Figure~\ref{fig:lena_hsvd_maps} shows the representation itself for the clean image and the first experimental corruption level. The HSVD panels use one common fixed scale from 0 to 1; no per-map min--max normalization is applied. The $5\%$ perturbation already produces widespread local changes in the entropy map, providing a direct visual counterpart to the strong initial HSVD-SSIM response.

\begin{figure}[H]
    \centering
    \includegraphics[width=0.88\textwidth]{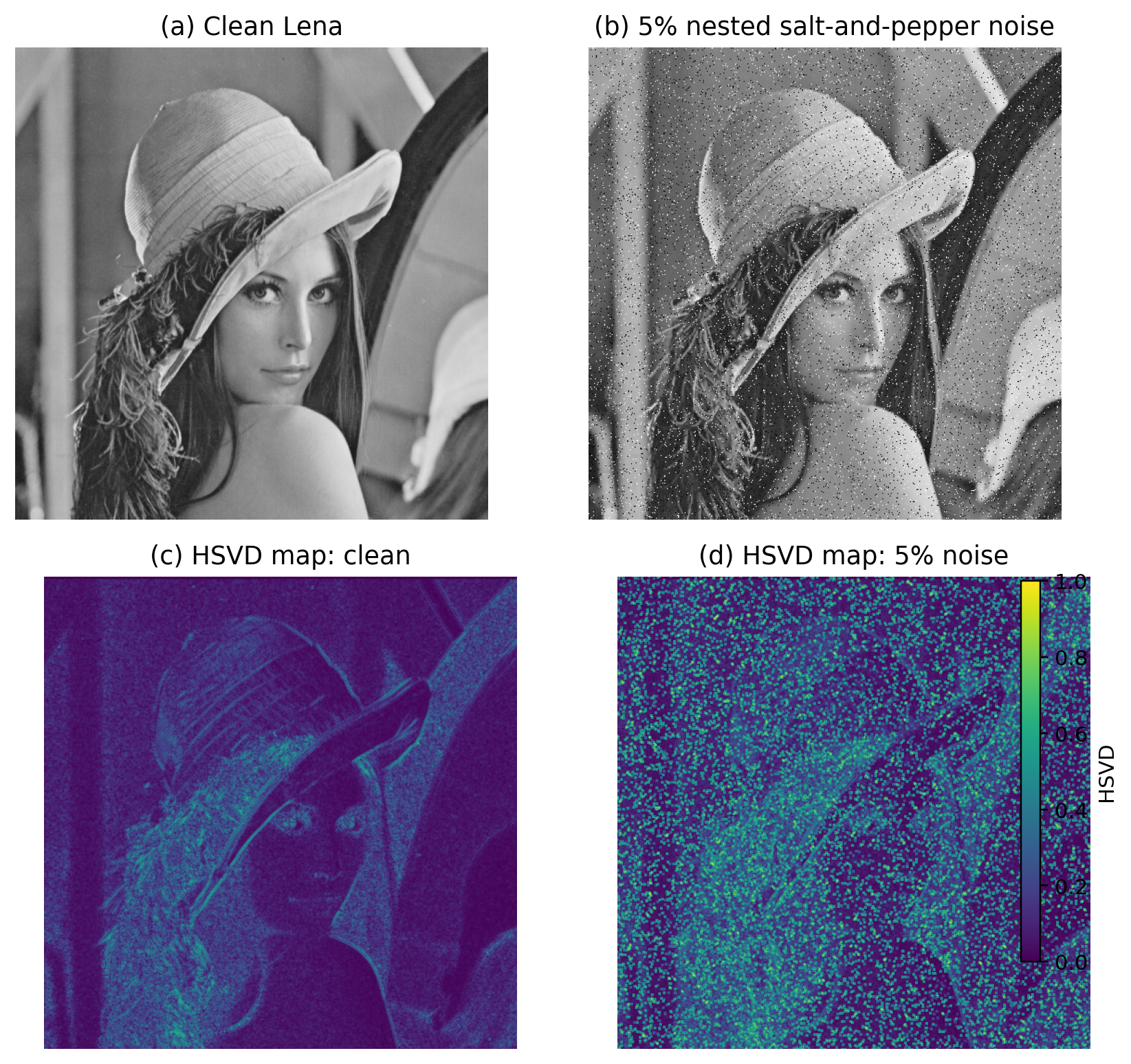}
    \caption{Direct visualization of the HSVD representation for the clean Lena image and the first nested-noise level ($5\%$). The two HSVD maps are shown on the same fixed $[0,1]$ scale and were not independently rescaled for display. The $5\%$ level is used because it is the first level in the experimental sequence.}
    \label{fig:lena_hsvd_maps}
\end{figure}

Figure~\ref{fig:lena_visual} shows four representative states from the same nested salt-and-pepper experiment. The $50\%$ image is strongly degraded but the main scene remains identifiable. At $80\%$, recognizable content is severely reduced, and at $100\%$ every pixel has been replaced by salt or pepper. These images are used only to provide visual context for the metric curves; no human opinion scores were collected for this illustrative experiment.

\begin{figure}[H]
    \centering
    \begin{subfigure}[b]{0.235\textwidth}
        \includegraphics[width=\textwidth]{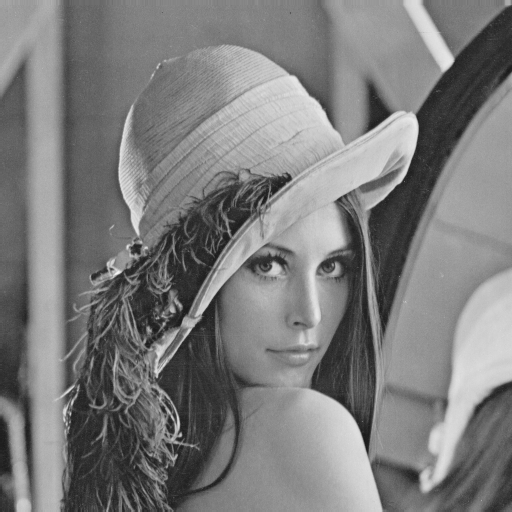}
        \caption{Clean}
    \end{subfigure}\hfill
    \begin{subfigure}[b]{0.235\textwidth}
        \includegraphics[width=\textwidth]{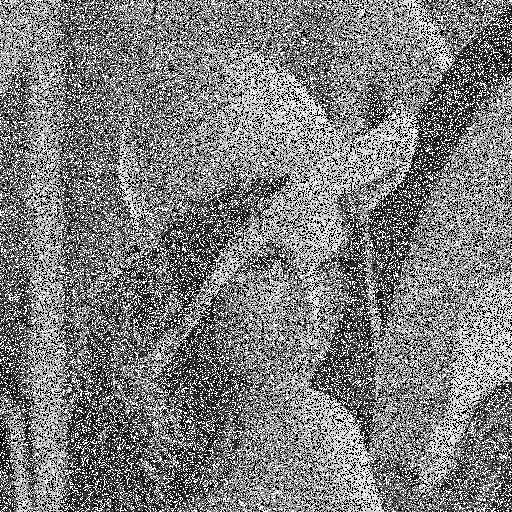}
        \caption{$50\%$ noise}
    \end{subfigure}\hfill
    \begin{subfigure}[b]{0.235\textwidth}
        \includegraphics[width=\textwidth]{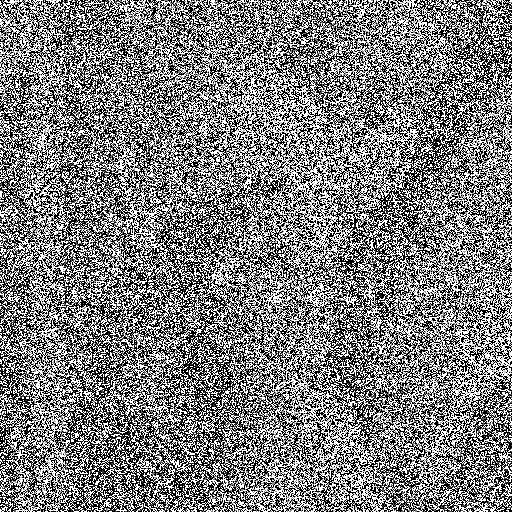}
        \caption{$80\%$ noise}
    \end{subfigure}\hfill
    \begin{subfigure}[b]{0.235\textwidth}
        \includegraphics[width=\textwidth]{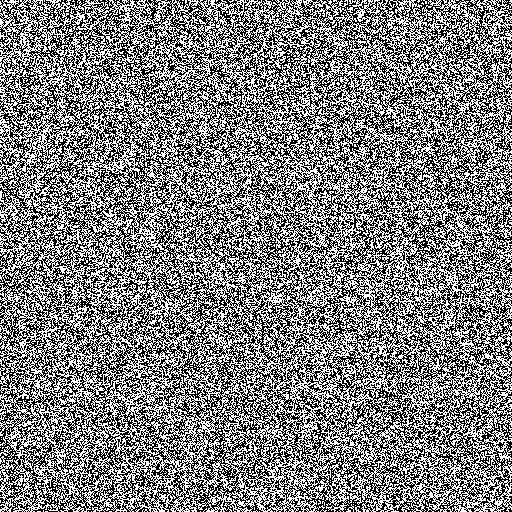}
        \caption{$100\%$ noise}
    \end{subfigure}
    \caption{Representative stages of the nested salt-and-pepper-noise experiment. Every higher noise level contains all corrupted pixels from the preceding level plus newly corrupted pixels.}
    \label{fig:lena_visual}
\end{figure}

The clean-reference comparison in Fig.~\ref{fig:lena_clean_curve} shows that the two representations behave differently from the first noise increment. At $5\%$ corruption, conventional SSIM falls from 1 to $0.334$, while HSVD-SSIM falls more sharply to $0.137$. At $20\%$, the corresponding values are $0.094$ and $0.032$; at $50\%$, they are $0.031$ and $0.013$. Thus HSVD-SSIM is more aggressive in declaring local structural departure from the clean reference. Conventional SSIM stays above HSVD-SSIM over nearly the entire noise range and decreases more gradually.

This behavior is important for interpretation. The clean-reference Lena result does \emph{not} indicate that HSVD-SSIM is a better absolute perceptual similarity measure. In fact, visually recognizable large-scale structure persists at noise levels for which both clean-reference scores are already close to zero. Rather, the steeper HSVD response indicates high sensitivity to the widespread local changes introduced by impulse corruption. This distinction motivated the separate neighboring-state analysis.

\begin{figure}[H]
    \centering
    \includegraphics[width=0.80\textwidth]{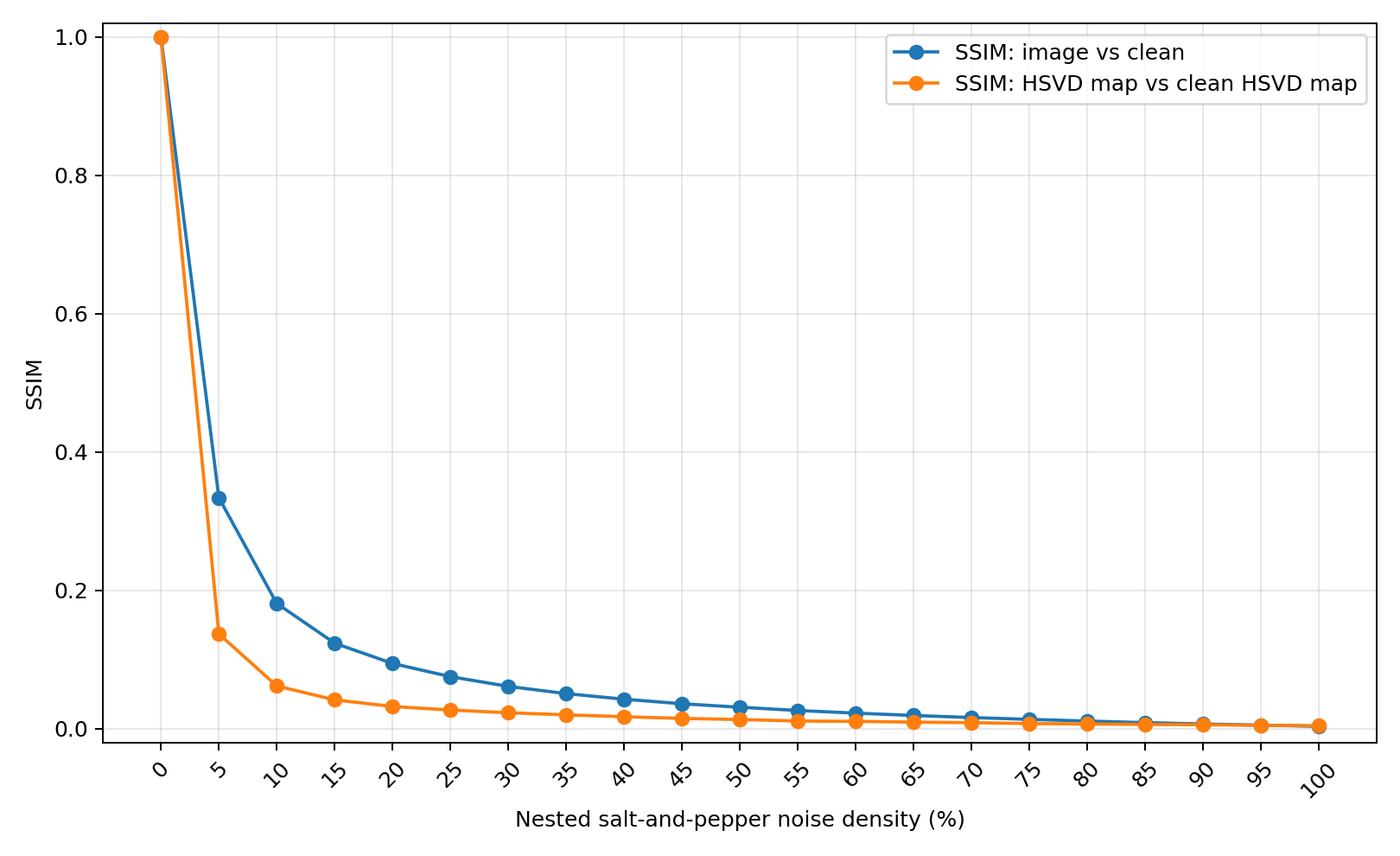}
    \caption{Similarity to the clean Lena reference as nested salt-and-pepper density increases. Conventional image-domain SSIM decreases more gradually, whereas HSVD-SSIM reacts more strongly to early local corruption.}
    \label{fig:lena_clean_curve}
\end{figure}

When consecutive noisy states are compared, the interpretation reverses in a useful way. Figure~\ref{fig:lena_neighbor_curve} shows that image-domain SSIM rapidly approaches one as the overall noise density becomes large: the 75--80\% pair has $SSIM=0.964$, and the 95--100\% pair reaches $0.970$. In other words, once both images are already heavily corrupted, the addition of a further $5\%$ corrupted pixels produces only a small image-domain SSIM change. HSVD-SSIM remains substantially lower: $0.873$ for 75--80\% and $0.902$ for 95--100\%. Across the five high-noise transitions ending at 80, 85, 90, 95, and 100\%, mean image-domain dissimilarity $1-SSIM$ is $0.0330$, whereas mean HSVD-domain dissimilarity is $0.1131$, approximately $3.42$ times larger.

Therefore the same property that makes HSVD-SSIM appear too sensitive in the clean-reference view becomes advantageous when the task is to detect an additional local structural change between two already degraded states. This is the physical motivation for the KADID neighboring-level and complementarity experiments.

\begin{figure}[H]
    \centering
    \includegraphics[width=0.80\textwidth]{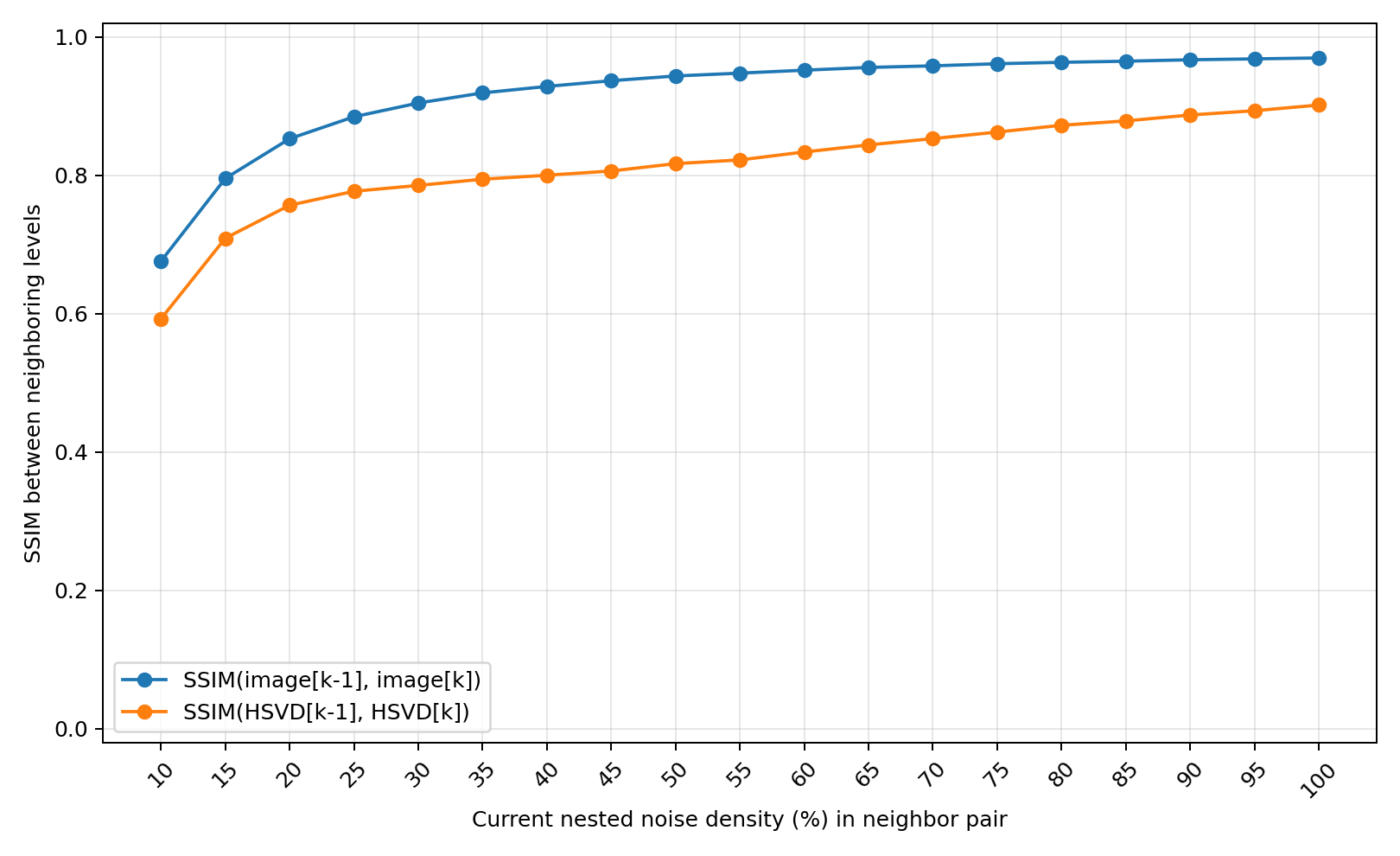}
    \caption{SSIM between consecutive nested-noise states. Image-domain SSIM approaches saturation near one in the high-noise regime, while HSVD-SSIM retains a substantially larger response to each additional $5\%$ corruption step.}
    \label{fig:lena_neighbor_curve}
\end{figure}

\begin{table}[H]
\centering
\caption{Selected values from the Lena experiment. ``Clean'' denotes similarity to the original image; ``neighbor'' denotes similarity to the immediately preceding noise level.}
\label{tab:lena_selected}
\small
\begin{tabular}{ccccc}
\toprule
Noise (\%) & SSIM-clean & HSVD-clean & SSIM-neighbor & HSVD-neighbor \\
\midrule
5   & 0.3341 & 0.1372 & 0.3341 & 0.1372 \\
20  & 0.0945 & 0.0323 & 0.8532 & 0.7570 \\
50  & 0.0312 & 0.0134 & 0.9439 & 0.8172 \\
80  & 0.0112 & 0.0070 & 0.9637 & 0.8726 \\
100 & 0.0037 & 0.0044 & 0.9700 & 0.9019 \\
\bottomrule
\end{tabular}
\end{table}

\subsection{KADID-10k Study I: HSVD-SSIM is not a standalone replacement for SSIM}
The first large-scale result provides an essential negative control. Across all 10,125 distorted KADID images, conventional SSIM correlated substantially better with subjective quality than HSVD-SSIM. The global SRCC values were $0.619$ and $0.450$, respectively; PLCC values were $0.576$ and $0.438$, and Kendall correlations were $0.447$ and $0.315$. Conventional SSIM also achieved the higher distortion-specific SRCC for 23 of the 25 distortion types.

\begin{table}[H]
\centering
\caption{Standalone full-reference correlation with KADID-10k subjective quality.}
\label{tab:kadid_standalone}
\begin{tabular}{lccc}
\toprule
Metric & SRCC & PLCC & Kendall \\
\midrule
Image-domain SSIM & \textbf{0.6189} & \textbf{0.5759} & \textbf{0.4470} \\
HSVD-SSIM & 0.4504 & 0.4377 & 0.3153 \\
\bottomrule
\end{tabular}
\end{table}

This result establishes the correct role of the proposed representation: HSVD-SSIM should not be marketed as a universally superior FR IQA metric. The question is instead whether the information it captures is complementary to information already represented by conventional SSIM.

\subsection{KADID-10k Study II: neighboring degradation levels}
For the 8,100 adjacent-level pairs, both metric-change measures had relatively low correlation with $|\Delta DMOS|$, confirming that the perceptual magnitude of a transition cannot be explained by one pairwise similarity value alone. Nevertheless, HSVD change was consistently more associated with human change than image-domain SSIM change. SRCC increased from $0.1115$ to $0.1338$, Kendall correlation from $0.0720$ to $0.0866$, and PLCC from approximately zero ($-0.0041$) to $0.0683$.

The SRCC difference of $+0.0223$ was stable under 1,000 reference-cluster bootstrap resamples: the 95\% confidence interval was $[+0.0134,+0.0304]$, and HSVD had the higher bootstrap SRCC in all resamples. Distortion-specific SRCC was higher for HSVD in 17 of the 25 distortion types. These absolute correlation values remain modest, so this experiment is interpreted as evidence of incremental sensitivity rather than as strong standalone perceptual prediction.

\begin{table}[H]
\centering
\caption{KADID neighboring-level correlation with the magnitude of subjective change $|\Delta DMOS|$. Metric change is $1-SSIM$ between adjacent distorted states.}
\label{tab:kadid_incremental}
\begin{tabular}{lccc}
\toprule
Metric change & SRCC & PLCC & Kendall \\
\midrule
Image-domain $1-SSIM$ & 0.1115 & -0.0041 & 0.0720 \\
HSVD-domain $1-SSIM$ & \textbf{0.1338} & \textbf{0.0683} & \textbf{0.0866} \\
\bottomrule
\end{tabular}
\end{table}

\subsection{Illustrative KADID cases with nearly matched conventional SSIM}
To make the complementarity effect visually interpretable, Fig.~\ref{fig:kadid_cases} presents two post-hoc illustrative examples selected from the same Experiment-I table. These examples are not used as statistical evidence; the aggregate benchmark results remain the primary test. Each row contains one reference image and two different distorted versions whose conventional SSIM values are almost identical, while both the subjective score and HSVD-SSIM differ substantially. All panels are shown in the same grayscale representation used in the quantitative analysis.

In the first example, the two distorted versions of reference \texttt{I22.png} have conventional SSIM values $0.934090$ and $0.934098$, a difference below $10^{-5}$. Their subjective scores are $3.13$ and $4.30$, whereas HSVD-SSIM changes from $0.6695$ to $0.9036$, correctly ordering the higher-subjective-quality state above the lower-quality state. In the second example, two versions of \texttt{I18.png} have SSIM values $0.933774$ and $0.933802$ but subjective scores $1.43$ and $3.74$; HSVD-SSIM changes from $0.5785$ to $0.8411$, again agreeing with the subjective ordering.

\begin{figure}[H]
\centering
\begin{subfigure}[b]{0.31\textwidth}
\includegraphics[width=\textwidth]{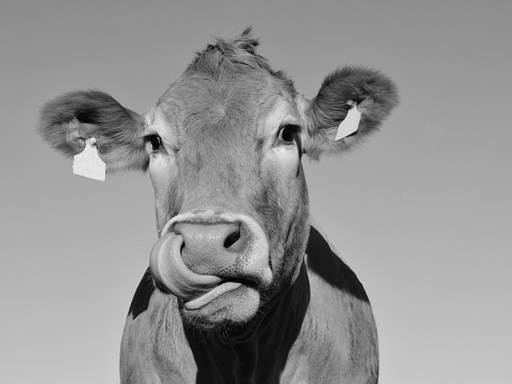}
\caption{I22 reference}
\end{subfigure}\hfill
\begin{subfigure}[b]{0.31\textwidth}
\includegraphics[width=\textwidth]{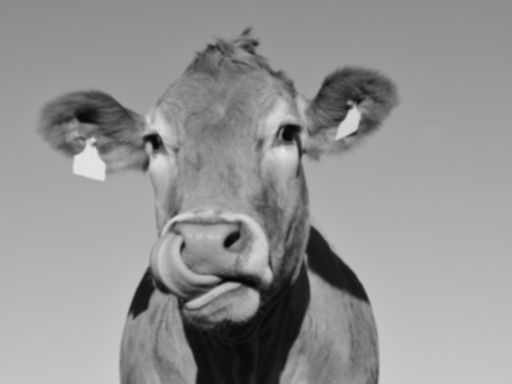}
\caption{I22 A: DMOS 3.13}
\end{subfigure}\hfill
\begin{subfigure}[b]{0.31\textwidth}
\includegraphics[width=\textwidth]{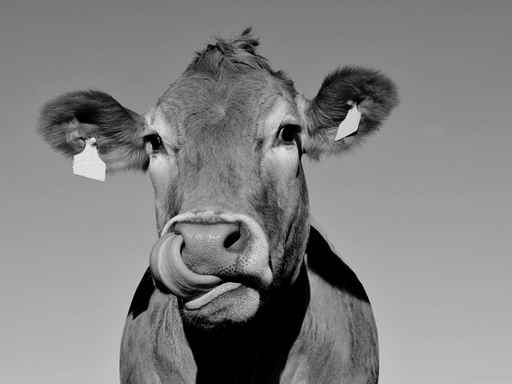}
\caption{I22 B: DMOS 4.30}
\end{subfigure}

\vspace{0.5em}

\begin{subfigure}[b]{0.31\textwidth}
\includegraphics[width=\textwidth]{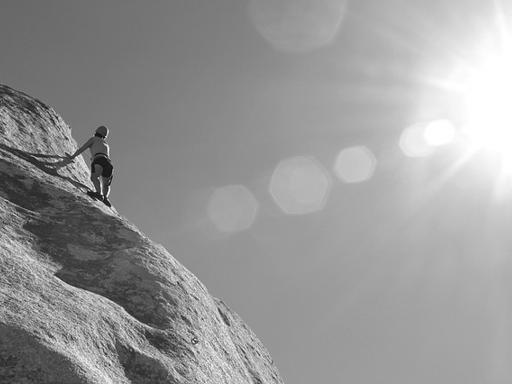}
\caption{I18 reference}
\end{subfigure}\hfill
\begin{subfigure}[b]{0.31\textwidth}
\includegraphics[width=\textwidth]{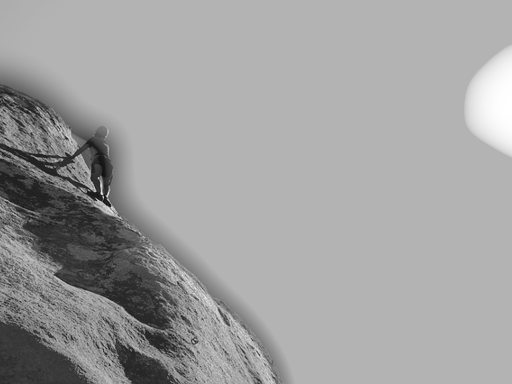}
\caption{I18 A: DMOS 1.43}
\end{subfigure}\hfill
\begin{subfigure}[b]{0.31\textwidth}
\includegraphics[width=\textwidth]{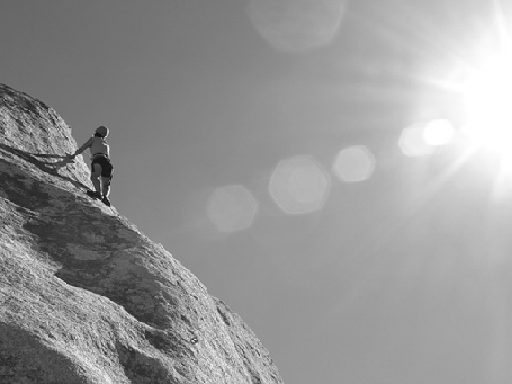}
\caption{I18 B: DMOS 3.74}
\end{subfigure}
\caption{Illustrative KADID-10k examples with nearly matched conventional SSIM but substantially different subjective quality. Top row: $SSIM_A=0.934090$, $SSIM_B=0.934098$, while $HSVD_A=0.6695$ and $HSVD_B=0.9036$. Bottom row: $SSIM_A=0.933774$, $SSIM_B=0.933802$, while $HSVD_A=0.5785$ and $HSVD_B=0.8411$. These examples illustrate complementarity and were selected post hoc; statistical conclusions are based on the full benchmark.}
\label{fig:kadid_cases}
\end{figure}

\subsection{KADID-10k Study III: full-reference complementarity}
The key FR experiment tested whether HSVD-SSIM improves prediction when conventional SSIM is already available. Table~\ref{tab:fr_complementarity} reports held-out predictions averaged across repeated grouped folds. SSIM-only reproduced the expected SRCC of approximately $0.618$. HSVD-only remained weaker ($0.449$). However, the joint model reached SRCC $0.659$, an absolute gain of $0.0414$ over SSIM-only. The combined model also improved Kendall correlation, PLCC, RMSE, MAE, and $R^2$, although the improvements in linear-error measures were smaller than the SRCC improvement.

\begin{table}[H]
\centering
\caption{Full-reference DMOS prediction under 10$\times$5-fold grouped cross-validation by reference image.}
\label{tab:fr_complementarity}
\small
\begin{tabular}{lcccccc}
\toprule
Model & SRCC & PLCC & Kendall & RMSE & MAE & $R^2$ \\
\midrule
SSIM only & 0.6175 & 0.5752 & 0.4455 & 0.8856 & 0.7601 & 0.3308 \\
HSVD only & 0.4492 & 0.4369 & 0.3142 & 0.9738 & 0.8394 & 0.1909 \\
SSIM + HSVD & \textbf{0.6590} & \textbf{0.5796} & \textbf{0.4744} & \textbf{0.8823} & \textbf{0.7582} & \textbf{0.3359} \\
\bottomrule
\end{tabular}
\end{table}

Reference-level bootstrap analysis confirmed that the ranking improvement was highly stable: the 95\% interval for $\Delta SRCC=SRCC_{\mathrm{SSIM+HSVD}}-SRCC_{\mathrm{SSIM}}$ was $[+0.0363,+0.0462]$, and every one of 2,000 bootstrap samples had a positive SRCC difference. The two FR scores are strongly related, so HSVD is not an independent quality axis; instead it acts as a correction channel containing residual information not represented by SSIM alone. The full-data standardized linear coefficients were approximately $0.687$ for SSIM and $-0.133$ for HSVD. The negative conditional coefficient should not be interpreted as saying that lower HSVD similarity is intrinsically better; it arises only after controlling for the strongly correlated SSIM predictor.

\begin{figure}[H]
\centering
\includegraphics[width=0.78\textwidth]{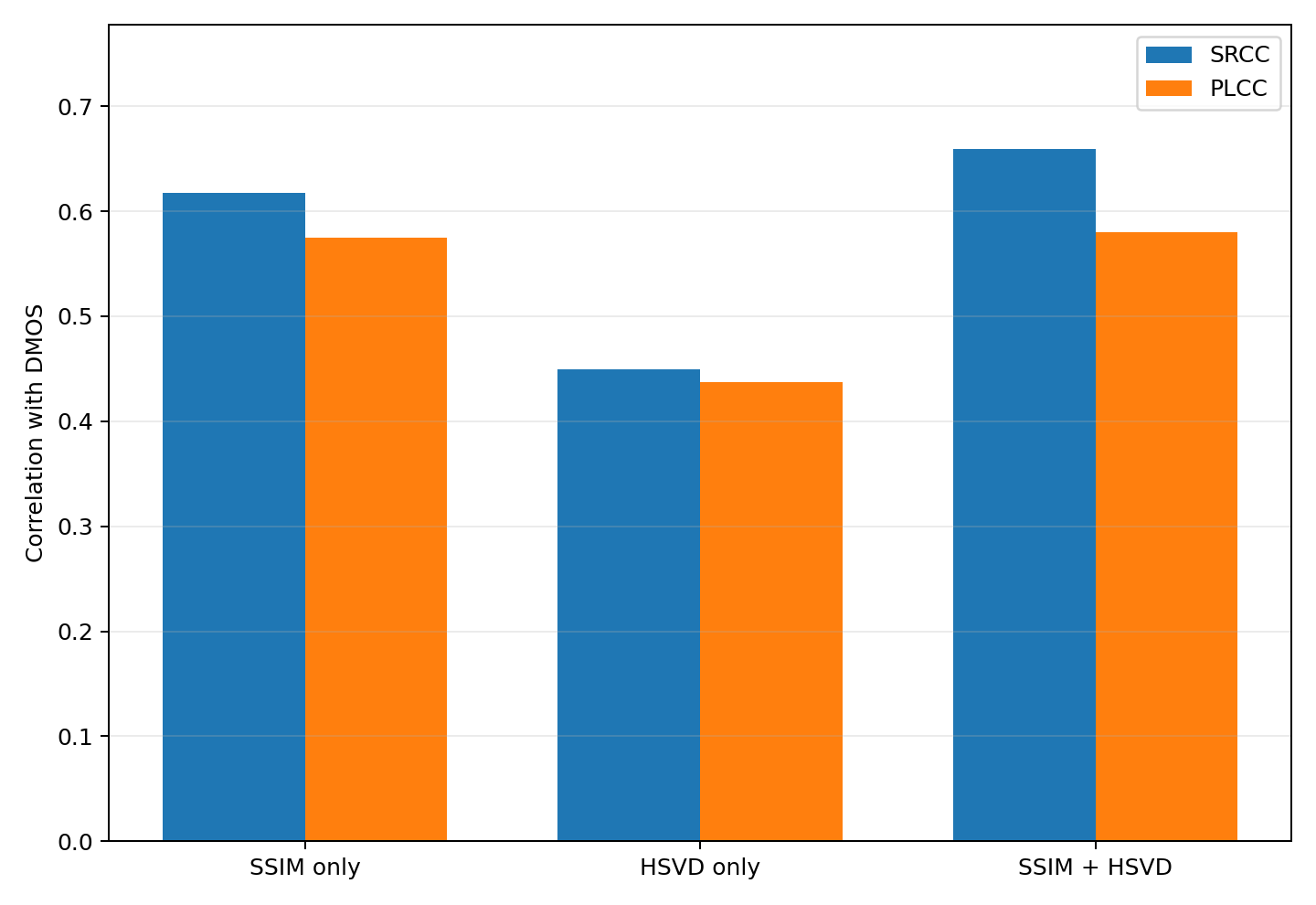}
\caption{Repeated grouped cross-validation for full-reference DMOS prediction. Adding HSVD-SSIM to conventional SSIM consistently improves rank correlation on held-out reference images.}
\label{fig:fr_cv}
\end{figure}

\subsection{KADID-10k Study IV: no-reference prediction}
The NR study asks a broader question: does an HSVD map contain useful information even when the pristine image is not available at all? Table~\ref{tab:nr_absolute} reports the out-of-fold performance of both model families. For HGB, Image-only features achieved SRCC $0.5277$, HSVD-only features achieved $0.5157$, and the Combined representation reached $0.5748$. The combined model simultaneously improved PLCC from $0.5643$ to $0.6067$, reduced RMSE from $0.9182$ to $0.8826$, and increased $R^2$ from $0.2807$ to $0.3353$. The same pattern was observed with Ridge regression: SRCC increased from $0.4157$ for Image-only to $0.4884$ for Combined features.

\begin{table}[H]
\centering
\caption{No-reference absolute DMOS prediction under 5$\times$5-fold grouped cross-validation by reference image. No pristine image, distortion identity, distortion type, or level is used as a predictive feature.}
\label{tab:nr_absolute}
\small
\begin{tabular}{llcccccc}
\toprule
Regressor & Features & SRCC & PLCC & Kendall & RMSE & MAE & $R^2$ \\
\midrule
HGB & Image & 0.5277 & 0.5643 & 0.3727 & 0.9182 & 0.7483 & 0.2807 \\
HGB & HSVD & 0.5157 & 0.5593 & 0.3627 & 0.9019 & 0.7397 & 0.3059 \\
HGB & Combined & \textbf{0.5748} & \textbf{0.6067} & \textbf{0.4109} & \textbf{0.8826} & \textbf{0.7075} & \textbf{0.3353} \\
\midrule
Ridge & Image & 0.4157 & 0.4220 & 0.2858 & 0.9859 & 0.8368 & 0.1708 \\
Ridge & HSVD & 0.4116 & 0.4554 & 0.2856 & 0.9640 & 0.8180 & 0.2071 \\
Ridge & Combined & \textbf{0.4884} & \textbf{0.5181} & \textbf{0.3409} & \textbf{0.9301} & \textbf{0.7748} & \textbf{0.2619} \\
\bottomrule
\end{tabular}
\end{table}

The HGB SRCC gain of approximately $+0.0471$ had a reference-bootstrap 95\% interval $[+0.0332,+0.0611]$; all 1,000 bootstrap differences were positive. For Ridge, the corresponding mean gain was approximately $+0.0721$ with 95\% interval $[+0.0515,+0.0915]$. For both model families, Combined SRCC exceeded Image-only SRCC in 17 of the 25 distortion types. This consistency across linear and nonlinear regressors supports the conclusion that HSVD statistics add complementary NR information rather than merely exploiting one particular estimator.

\begin{figure}[H]
\centering
\begin{subfigure}[b]{0.48\textwidth}
\includegraphics[width=\textwidth]{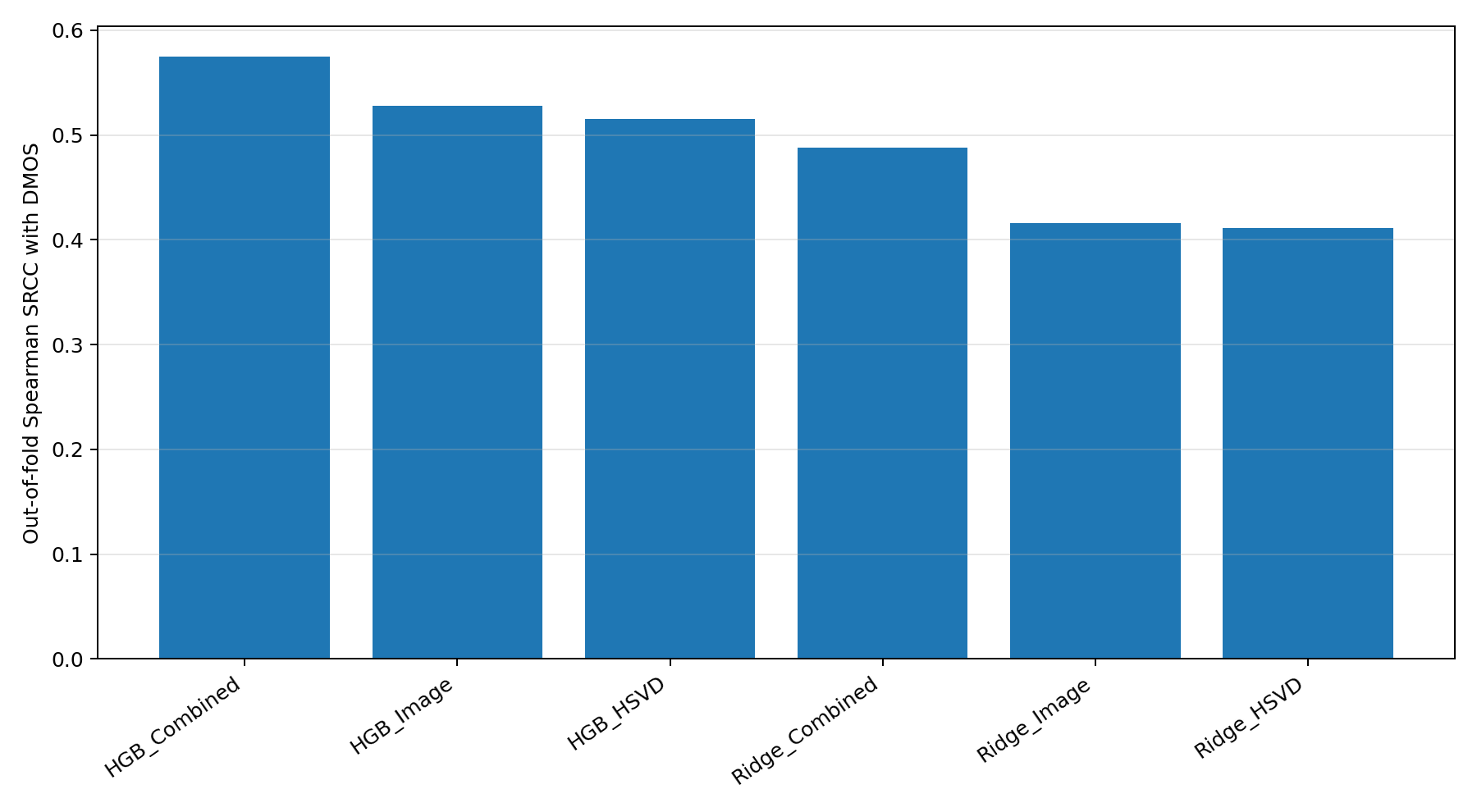}
\caption{Absolute DMOS SRCC}
\end{subfigure}\hfill
\begin{subfigure}[b]{0.48\textwidth}
\includegraphics[width=\textwidth]{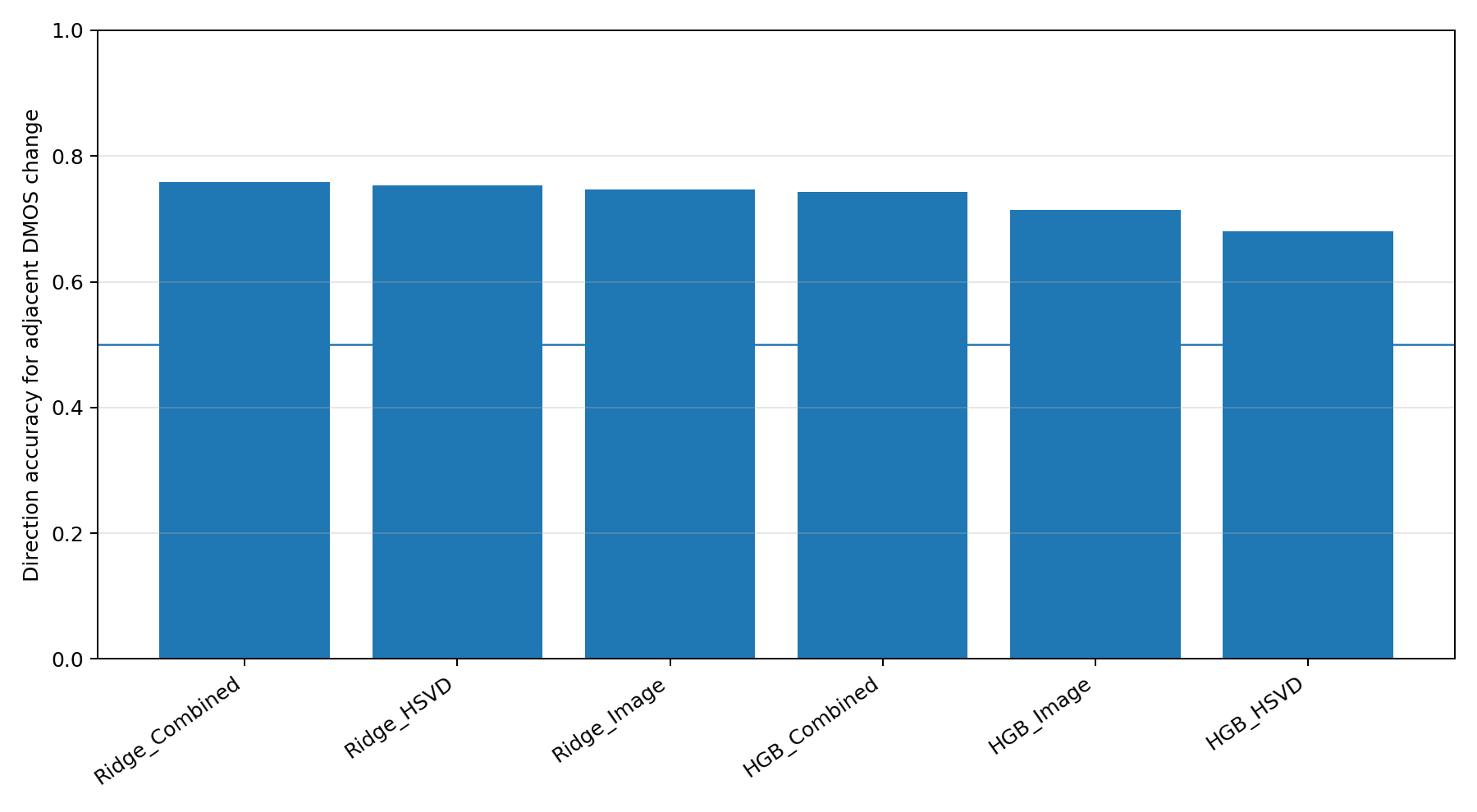}
\caption{Neighbor-pair ordering accuracy}
\end{subfigure}
\caption{No-reference performance. Left: absolute-quality rank correlation for Image-only, HSVD-only, and Combined feature sets. Right: accuracy of the quality ordering between adjacent states derived from independently predicted out-of-fold DMOS values.}
\label{fig:nr_summary}
\end{figure}

The out-of-fold absolute predictions were subsequently differenced to estimate quality change without any new fitting. For HGB, the Combined model achieved signed $\Delta DMOS$ SRCC $0.3597$, compared with $0.2664$ for Image-only. Correlation with $|\Delta DMOS|$ increased from $0.2533$ to $0.3385$, and pairwise direction accuracy increased from $0.7138$ to $0.7433$. For stronger subjective changes, $|\Delta DMOS|\ge0.25$, direction accuracy reached $0.7924$. These values are not high enough to describe the model as a complete NR-IQA solution, but they demonstrate that the additional HSVD information survives the more difficult setting in which no pristine reference is available.

\begin{figure}[H]
\centering
\includegraphics[width=0.72\textwidth]{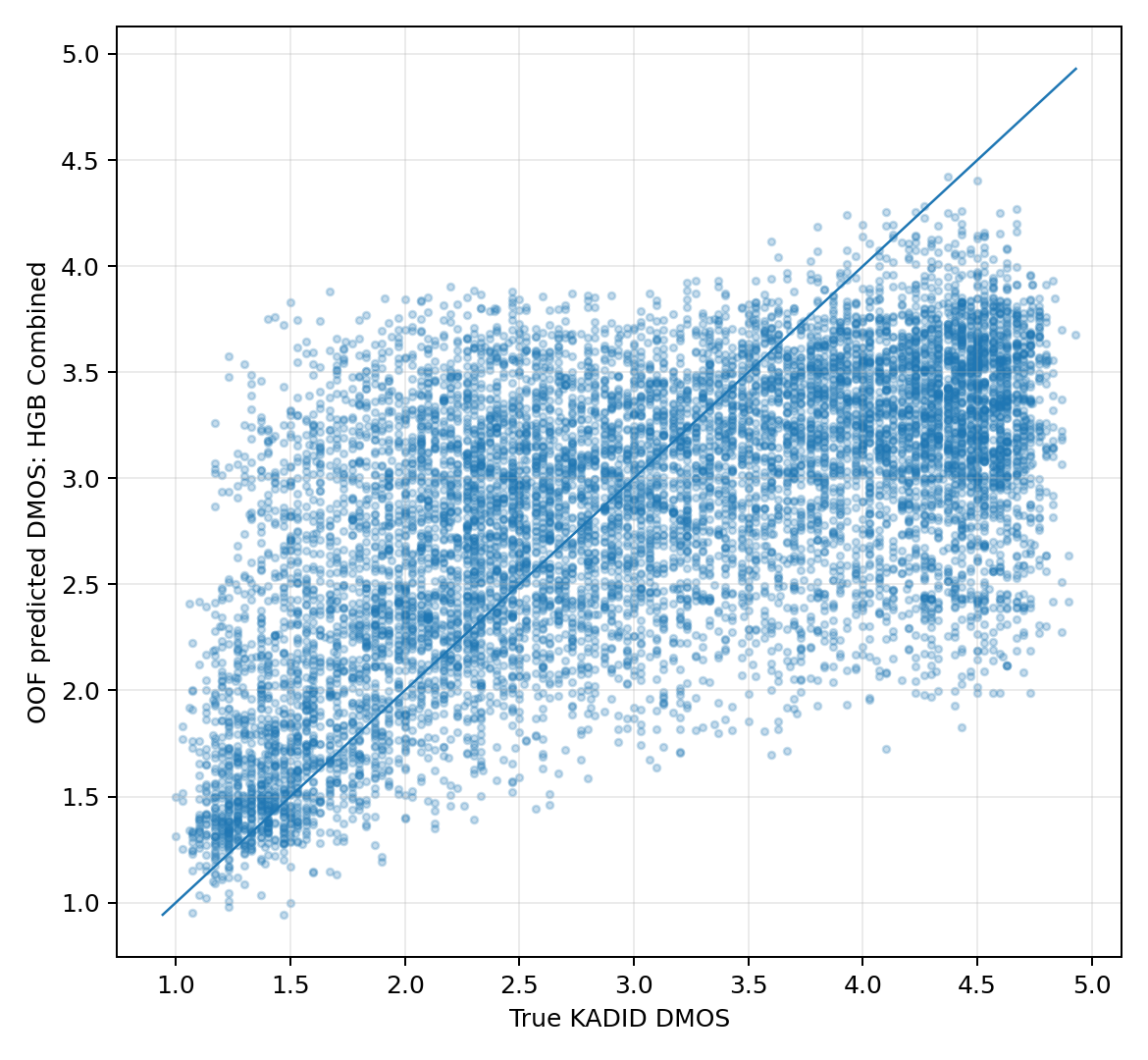}
\caption{True versus out-of-fold predicted DMOS for the HGB Combined no-reference model. Predictions exhibit noticeable regression toward the center of the quality range, highlighting the remaining difficulty of absolute no-reference quality estimation.}
\label{fig:nr_scatter}
\end{figure}

\section{Discussion}
\subsection{Sensitivity to change is not the same as absolute perceptual quality}
The experiments show that image-domain SSIM and HSVD-SSIM emphasize different aspects of degradation. This distinction is already visible in the controlled Lena experiment. When each state is compared with the clean image, HSVD-SSIM falls very rapidly, indicating that the local singular-value distribution is strongly perturbed even while recognizable large-scale image content remains. Conventional SSIM decreases more gradually and therefore behaves more naturally as a coarse clean-reference similarity trajectory in this example. When two already degraded neighboring states are compared, however, image-domain SSIM approaches one whereas HSVD-SSIM retains a substantially larger response to the newly corrupted pixels. In other words, a descriptor can be highly sensitive to an incremental structural change without being the best predictor of absolute human-rated quality.

This distinction has a direct analogue in subjective IQA methodology. Pairwise and boosted triplet comparisons are used precisely because fine quality differences may be difficult to resolve on an absolute rating scale, and restoration-oriented databases emphasize that perceptually relevant artifacts need not follow simple pointwise distortion models \cite{Men2021Pairwise,Gu2020PIPAL}. Our KADID neighboring-level experiment is therefore best interpreted as a sensitivity test rather than as an alternative absolute IQA benchmark. The low absolute correlations between either neighboring-state dissimilarity and $|\Delta DMOS|$ reinforce this conservative interpretation: HSVD retains more response to incremental degradation, but pairwise similarity alone is not sufficient to predict the magnitude of subjective quality change.

\subsection{What local SVD entropy measures}
The behavior of HSVD can be understood from the singular-value spectrum of a patch. If most matrix energy is concentrated in one singular direction, the normalized distribution approaches $p_1\approx1$ and $p_{k>1}\approx0$, so $H_{\mathrm{SVD}}$ approaches zero. If the singular energy is distributed more evenly across several directions, entropy increases. This connects HSVD to a broader SVD literature in which dominant singular components are associated with organized image structure and weaker components with fine-scale detail or noise, and in which singular-value entropy has been used as a compact measure of structural complexity \cite{Guo2016SVDDenoising,Buisine2021SVDEntropy}. Importantly, this interpretation does not imply that every degradation must monotonically increase local entropy; blur, structured artifacts, quantization, and noise can alter the spectrum in different ways. What the descriptor records is a change in the \emph{shape} of the local singular-value distribution.

The normalization of singular values makes this distinction especially clear. Multiplying a nonzero patch by a scalar leaves $H_{\mathrm{SVD}}$ unchanged, whereas conventional SSIM explicitly contains luminance and contrast terms. HSVD therefore suppresses one class of amplitude information while retaining information about the relative organization of local matrix components. Such invariance is advantageous when the desired cue is structural complexity, but it also explains why HSVD alone cannot represent all perceptually relevant aspects of quality. Its strongest role is consequently as an additional channel rather than a complete quality coordinate.

\subsection{Full-reference complementarity rather than replacement}
The KADID full-reference results support this interpretation quantitatively. HSVD-SSIM alone is substantially weaker than conventional SSIM, but the combined model improves ranking from SRCC $0.618$ to $0.659$. The gain in PLCC and RMSE is much smaller than the gain in SRCC. This pattern suggests that HSVD is particularly useful for correcting local ordering errors---cases in which conventional SSIM assigns very similar scores to images that human observers rank differently. The matched-SSIM examples in Fig.~\ref{fig:kadid_cases} provide a visual illustration of this effect, while the repeated grouped cross-validation and reference-level bootstrap show that the aggregate improvement is not dependent on a small number of selected examples.

This result fits the broader IQA literature more naturally as a complementarity finding than as a new standalone-metric claim. FSIM combines phase congruency and gradient magnitude because the two cues respond to different perceptual properties, while DISTS explicitly combines structure and texture statistics in a learned feature space \cite{Zhang2011FSIM,Ding2020DISTS}. Multi-domain and optimization-based IQA methods similarly exploit partially non-redundant information from distinct representations. In the present linear model, the negative conditional HSVD coefficient should therefore be viewed as a suppression/correction effect in the presence of a strongly correlated SSIM predictor, not as evidence that lower HSVD similarity is intrinsically associated with better quality. The statistically stable improvement is the relevant finding: HSVD contains residual ordering information after the conventional structural score is already known.

\subsection{Distortion dependence and the limits of a single structural cue}
Complementarity is not uniform across distortion classes. HSVD produced the stronger neighboring-level correlation in 17 of the 25 KADID distortion types, and the combined no-reference models likewise improved over Image-only features in 17 of 25 types. This heterogeneity is expected. A local singular-value distribution is directly sensitive to changes in local rank, orientation mixture, texture, edge organization, and fine-scale irregularity, whereas other distortions act more strongly through luminance, color, global geometry, or semantic appearance. A single scalar local-complexity mechanism should therefore not be expected to respond optimally to every distortion family.

This observation also helps reconcile the apparently contradictory standalone and fusion results. If HSVD encoded essentially the same information as SSIM, adding it would not provide a reproducible gain. If it encoded a universally superior quality variable, it would outperform SSIM alone. The observed intermediate case---weaker standalone performance but consistent fused improvement---is precisely what is expected from a descriptor that is informative for a subset of perceptual failure modes and partly redundant for others. A more detailed distortion-wise analysis, particularly for blur, impulse/noise, pixelation, texture-related, and chromatic distortions, is therefore an important next step.

\subsection{From traversal-designed entropy maps to direct 2-D matrix entropy}
The direct formulation also clarifies the methodological relationship to our earlier two-dimensional entropy work. In \cite{Velichko2022EntropyApprox}, local circular kernels were converted to specially ordered one-dimensional sequences before SVD, sample, permutation, or neural-network entropy was evaluated. The circular traversal was designed to reduce orientation sensitivity of the resulting two-dimensional entropy distributions. That approach was useful because it allowed one-dimensional entropy algorithms to be transferred to images, but it necessarily introduced a sequence-construction rule and, for SVD entropy, a subsequent embedding stage.

The present method removes both steps. A $3\times3$ patch is treated directly as a matrix, so the descriptor depends only on its singular values. Right-angle rotation and reflection invariance then follows exactly from SVD rather than from a traversal designed to make the resulting sequence approximately stable. The direct formulation is also algorithmically simpler because no ordered series or delay-embedding matrix must be constructed. We deliberately do not claim a hardware-independent speedup over the earlier implementation because the two studies were executed under different computational conditions; a controlled benchmark of circular-kernel SvdEn2D against direct-patch HSVD would be the appropriate way to quantify that advantage.

The invariance should also be interpreted locally. Rotating or reflecting a patch does not change its HSVD value, but rotating only one of two complete images destroys pixel correspondence and therefore does not make HSVD-SSIM globally registration invariant. Moreover, orientation invariance necessarily discards some orientation-specific information. That loss is another reason why the descriptor is most naturally used together with image-domain or other directional features.

\subsection{Implications for no-reference IQA}
The no-reference experiment provides an independent test of the same hypothesis. Once the pristine image is removed, HSVD-SSIM itself is unavailable; only statistics of the observed HSVD map can be used. Nevertheless, adding these statistics improves both Ridge and HGB models. For HGB, SRCC increases from $0.528$ to $0.575$, PLCC from $0.564$ to $0.607$, and RMSE decreases from $0.918$ to $0.883$. Thus, unlike the full-reference fusion experiment where the most pronounced gain is in ranking, the no-reference result also shows a clearer improvement in numerical prediction error.

This result is consistent with interpretable NR-IQA research that combines natural-scene statistics, texture, gradient, entropy, transform-domain, and singular-value-related features because these families describe different departures from natural image structure \cite{Varga2020NRFusion,Cui2020DualDomainNR,Varga2021GlobalStats,Varga2023GlobalLocal}. Thus, the NR experiment supports the narrower claim that HSVD carries quality-related information without a reference; it is not intended as a state-of-the-art blind-IQA benchmark against large contemporary deep networks.

The independently predicted no-reference scores also provide a practical direction for before/after monitoring. Given two processed states $A$ and $B$, a single-image model can produce $\widehat Q(A)$ and $\widehat Q(B)$ and use their difference to estimate whether quality is likely to have improved or deteriorated. The present KADID result is only a surrogate validation of this idea because adjacent distortion levels are not actual restoration trajectories, but it establishes the feasibility of testing such a workflow without access to a pristine image.

\subsection{Potential applications}
The most immediate potential application is quality monitoring during image restoration, denoising, or enhancement. In a full-reference setting, HSVD could serve as an additional structural channel when selecting algorithms or parameters that receive nearly indistinguishable conventional SSIM scores. In a no-reference setting, HSVD statistics could contribute to a learned quality estimator used to compare states before and after processing. The latter is particularly relevant to real-world pipelines where the original pristine scene is unavailable.

A second application is iterative stopping or convergence monitoring. SVD entropy has previously been used as a stopping indicator for computer-generated image rendering \cite{Buisine2021SVDEntropy}. The larger neighboring-state response observed here suggests that local HSVD maps may be useful for detecting when iterative restoration is still materially changing local structure and when additional iterations yield diminishing changes. This hypothesis requires direct validation because structural change is not synonymous with perceptual improvement.

Other candidate domains include remote-sensing preprocessing and change-quality monitoring, where spatial entropy maps have already been used as irregularity descriptors \cite{Velichko2022EntropyApprox}; scientific and medical imaging, where interpretable local structure measures can complement application-specific fidelity metrics; and compression or transcoding quality control. Color-aware extensions may also exploit quaternion-SVD or luminance/chrominance decompositions, since SVD-derived color IQA descriptors have already been investigated \cite{Sang2020QuaternionSVD}. These applications should be regarded as directions for validation rather than as demonstrated use cases of the present study.

\begin{table}[H]
\centering
\caption{Interpretation of the main empirical observations. The third column states conclusions that should \emph{not} be inferred from the corresponding result.}
\label{tab:interpretation}
\small
\begin{tabular}{p{0.29\textwidth}p{0.34\textwidth}p{0.29\textwidth}}
\toprule
Observation & Supported interpretation & Does not imply \\
\midrule
HSVD falls faster than SSIM relative to clean Lena & Strong response to widespread local structural perturbation & Better absolute perceptual-quality prediction \\
HSVD retains larger response between heavily degraded neighbors & Greater incremental structural sensitivity in that regime & Accurate prediction of $|\Delta DMOS|$ by similarity alone \\
Standalone HSVD $<$ standalone SSIM on KADID & Local spectral complexity is insufficient as a universal FR quality variable & HSVD contains no useful quality information \\
SSIM+HSVD $>$ SSIM in grouped CV & HSVD contains residual information complementary to SSIM & The two scores are statistically independent quality axes \\
Image+HSVD $>$ Image-only in NR prediction & HSVD contains reference-free quality-related information & State-of-the-art blind IQA performance \\
\bottomrule
\end{tabular}
\end{table}
\section{Limitations and future work}
The present study is intentionally focused on the information value of one simple local representation, and several limitations define the next stage of the work. First, all reported HSVD maps use a fixed $3\times3$ window. Window size determines the balance between locality and structural averaging: smaller neighborhoods emphasize sharp local transitions, whereas larger neighborhoods integrate broader texture and shape. Our earlier circular-kernel entropy study showed analogous scale-dependent changes in the sharpness of two-dimensional entropy distributions \cite{Velichko2022EntropyApprox}. A systematic multiscale extension, for example combining $E_3(X)$, $E_5(X)$, and $E_7(X)$, is therefore a high-priority ablation rather than a cosmetic parameter search.

Second, the current KADID analysis is grayscale. This choice ensures that conventional SSIM and HSVD operate on exactly the same luminance representation, but it necessarily reduces sensitivity to distortions that are primarily chromatic. Future work should evaluate luminance/chrominance HSVD channels, perceptually motivated color spaces, or quaternion-SVD formulations. Existing quaternion-SVD IQA work provides a useful reference point for testing whether inter-channel singular structure adds information beyond the grayscale descriptor \cite{Sang2020QuaternionSVD}.

Third, KADID-10k contains controlled artificial distortions. Its size and subjective annotations make it appropriate for the present hypothesis tests, but generalization to authentic mixed distortions has not yet been established. Cross-database evaluation on natural in-the-wild IQA data, together with restoration-oriented datasets such as PIPAL \cite{Gu2020PIPAL}, is needed before broad claims about perceptual generalization can be made. The same issue applies to the neighboring-level experiment: KADID severity levels are a useful surrogate for successive degradation states, not actual outputs from a denoising or restoration trajectory.

Fourth, the present no-reference models are deliberately modest. Ridge and HGB were used to test whether HSVD adds information to interpretable image statistics, not to compete with the strongest contemporary deep BIQA systems. A more complete future comparison should include established blind-IQA baselines and modern learned representations while preserving reference-grouped or cross-database validation. Such experiments would determine whether HSVD remains useful when stronger feature families are already present, rather than only when combined with simple statistics.

Fifth, dense SVD-map construction is computationally more expensive than conventional SSIM. In the current unoptimized Python implementation, HSVD construction required approximately $0.33$~s per distorted KADID image. The $3\times3$ case is small enough to permit substantial optimization through vectorization, batched linear algebra, GPU execution, lookup/approximation strategies, or analytic treatment of very small matrices. The direct method is algorithmically simpler than the earlier circular-sequence-plus-embedding formulation, but a controlled timing study on identical hardware is still required before claiming a quantitative speed advantage.

Sixth, the invariance established in Section~2 is local rather than global. HSVD is exactly invariant to right-angle rotation and reflection of an individual patch and to uniform multiplicative scaling of a nonzero patch, but HSVD-SSIM still assumes spatial registration of the two images. Orientation invariance can also suppress potentially useful directional information. Future variants could therefore combine HSVD with gradient orientation or anisotropic descriptors instead of treating invariance as universally desirable.

Seventh, the distortion-wise results are heterogeneous. The aggregate gains are statistically stable, yet not every KADID distortion type benefits from HSVD. This should be investigated through class-wise ablation and feature attribution, especially for blur, noise, pixelation, texture-related, geometric, and chromatic distortions. Such analysis may support an adaptive fusion strategy in which HSVD receives greater weight only when the observed degradation is compatible with the local-complexity cue.

Finally, the two KADID case studies in Fig.~\ref{fig:kadid_cases} were selected after the aggregate analysis to make the complementarity phenomenon visually interpretable. They are illustrations rather than independent validation. The statistical conclusions rely on the complete 10,125-image dataset, grouped cross-validation, and reference-level bootstrap analysis. A particularly important future experiment is therefore prospective: generate real noisy/restored image trajectories with multiple restoration algorithms and parameter settings, collect or use independent subjective ratings, and test whether FR and NR HSVD-based models correctly identify both the direction and magnitude of perceptual improvement. This would directly connect the present representation study to practical restoration monitoring and stopping criteria.

\section{Conclusion}
This work examined direct local SVD entropy as an interpretable structural representation for image quality assessment. The resulting HSVD map is computed from valid two-dimensional patches without flattening or delay embedding, and its local values are exactly invariant to right-angle rotation and reflection. The experiments show that this representation should be understood as a complementary spectral-complexity channel rather than as a replacement for conventional SSIM.

The controlled Lena experiment separated absolute clean-reference similarity from incremental change sensitivity: HSVD reacts more strongly to local corruption and retains a larger response between already heavily degraded neighboring states. KADID-10k then established the perceptual limits and value of that sensitivity. HSVD-SSIM alone was weaker than SSIM for absolute FR prediction, but adding HSVD to SSIM produced a stable increase in grouped cross-validated ranking accuracy. The same complementary effect persisted without a pristine reference, where HSVD-derived single-image statistics improved both linear and nonlinear quality-prediction models and improved the reconstruction of quality changes between neighboring distortion states.

Taken together, the results support a conservative but useful conclusion: local singular-value entropy captures a component of degradation that overlaps with, but is not exhausted by, conventional image-domain structural similarity. Its most promising role is therefore in fused FR/NR assessment, fine-grained monitoring of image-processing changes, and future multiscale or color-aware quality models. Validation on authentic distortions and real restoration trajectories, together with computational optimization and stronger NR baselines, is required before these applications can be considered established.

\section*{CRediT authorship contribution statement}
\textbf{Andrei Velichko:} Conceptualization, Methodology, Software, Formal analysis, Visualization, Writing - original draft, Project administration, Funding acquisition. \textbf{Petr Boriskov:} Validation, Discussion, Writing - original draft, Writing - review and editing.

\section*{Declaration of competing interest}
The authors declare that they have no known competing financial interests or personal relationships that could have appeared to influence the work reported in this paper.

\section*{Data and code availability}
The data and Python code supporting the findings of this study are available from the corresponding author upon reasonable request.

\section*{Funding}
This research was funded by the Russian Science Foundation, grant number 22-11-00055-P.

\printbibliography
\end{document}